\documentclass{sig-alternate-05-2015}

\newtheorem{definition}{Definition}
\newtheorem{property}{Property}
\newtheorem{thm}{Theorem}
\newtheorem{cor}{Corollary}
\newtheorem{lem}{Lemma}
\newtheorem{example}{Example}

\newcommand{\cy}{\mathcal{Y}}

\newcommand{\E}{\mathbb{E}}

\newcommand{\xmi}{{X^{-i}}}

\begin{document}

\CopyrightYear{2016} 
\setcopyright{acmlicensed}
\conferenceinfo{CCS'16,}{October 24 - 28, 2016, Vienna, Austria}
\isbn{978-1-4503-4139-4/16/10}\acmPrice{\$15.00}
\doi{http://dx.doi.org/10.1145/2976749.2978308}

\clubpenalty=10000 
\widowpenalty = 10000

\title{Differential Privacy as a Mutual Information Constraint}

\numberofauthors{2}

\author{
\alignauthor
Paul Cuff\\
       \affaddr{Princeton University}\\
\alignauthor
Lanqing Yu\\
       \affaddr{Princeton University}\\
}

\maketitle

\begin{abstract}
Differential privacy is a precise mathematical constraint meant to ensure privacy of individual pieces of information in a database even while queries are being answered about the aggregate.  Intuitively, one must come to terms with what differential privacy does and does not guarantee.  For example, the definition prevents a strong adversary who knows all but one entry in the database from further inferring about the last one.  This strong adversary assumption can be overlooked, resulting in misinterpretation of the privacy guarantee of differential privacy.

Herein we give an equivalent definition of privacy using mutual information that makes plain some of the subtleties of differential privacy.  The mutual-information differential privacy is in fact sandwiched between $\epsilon$-differential privacy and $(\epsilon,\delta)$-differential privacy in terms of its strength.  In contrast to previous works using unconditional mutual information, differential privacy is fundamentally related to conditional mutual information, accompanied by a maximization over the database distribution.  The conceptual advantage of using mutual information, aside from yielding a simpler and more intuitive definition of differential privacy, is that its properties are well understood.  Several properties of differential privacy are easily verified for the mutual information alternative, such as composition theorems.
\end{abstract}

%
%

%
%



\section{Introduction}

Differential privacy is a concept proposed in \cite{dwork:dp} for database privacy.  It allows queries to be answered about aggregate quantities of data while protecting the privacy of individual entries in the database.  In the absence of a precise mathematical framework such as differential privacy, practitioners have been tempted to use various rules-of-thumb to protect privacy (e.g. ``don't answer a query that averages fewer than $k$ entries together''---see the query restriction approach in \cite{Adam:1989:SMS:76894.76895}).  Instead, differential privacy directly addresses the statistical distinguishability of the database and has led to algorithms for answering general queries with just the right amount of randomness used in order to preserve privacy.\footnote{Differential privacy does not assume the adversary has any computational limitation.}

Differential privacy requires that two adjacent databases, which differ in only one entry, are statistically indistinguishable, as measured by a probabilistic metric defined in Section~\ref{sec:pre}.  This guarantee is particularly effective for making individuals feel comfortable contributing personal information to a dataset.  For instance, if a person decides to participate in a survey, his answers only constitute one response out of the entire collection, and the responses of other people remain unchanged.  Differential privacy is meant to assure the one participant that his answers are concealed.

This privacy metric has gained a lot of traction in recent years.  The main contribution of this work is to cast differential privacy as a mutual information constraint.  There have been many attempts in the literature to connect differential privacy to mutual information.  Here we give not only a connection but an equivalence.

To briefly summarize the main result, consider a database  $X^n = (X_1,\ldots,X_n)$ that returns a query response $Y$ according to a random mechanism $P_{Y|X^n}$.  Let $X^{-i}$ denote the set of database entries excluding $X_i$.

\begin{definition}[$(\epsilon, \delta)$-Differential Privacy \cite{dwork:epsilondelta}]
\label{def:eddp}
    A randomized mechanism $P_{Y|X^n}$ satisfies {\em $(\epsilon, \delta)$-differential privacy} if for all neighboring database instances $x^n$ and $\tilde{x}^n$
    \begin{equation}
    \label{eq:intro DP}
        P_{Y|X^n=x^n} \stackrel{(\epsilon, \delta)}{\approx} P_{Y|X^n=\tilde{x}^n},
    \end{equation}
    where the approximation in \eqref{eq:intro DP} is defined later in Definition~\ref{def:ed close}, and neighboring database instances are defined in Definition~\ref{def:neighbor} as any pair of database vectors that differ in only one entry (i.e. Hamming distance one).\footnote{Another similar definition for ``neighbor'' exists in the literature, involving the removal of one entry of the database.}
    \end{definition}

\begin{definition}[Mutual-Information Diff. Priv.]
\label{def:midp}
    A randomized mechanism $P_{Y|X^n}$ satisfies {\em $\epsilon$-mutual-information differential privacy} if \begin{equation}
    \label{eq:intro midp}
        \sup_{i,P_{X^n}} I( X_i;Y|X^{-i}) \leq \epsilon \text{ nats}.
    \end{equation}
    Note that {\em nats} are the information units that result from using the natural logarithm instead of the logarithm base two, which would give bits.
\end{definition}

The main claim of this paper, which appears in Section~\ref{section main result}, is an equivalence between mutual-information differential privacy (MI-DP) and the standard definition of $(\epsilon,\delta)$-differential privacy ($(\epsilon,\delta)$-DP).  The original definition of differential privacy \cite{dwork:dp}, defined formally in Section~\ref{sec:pre}, parameterized privacy with a single positive number $\epsilon$.  For various reasons it has since been relaxed \cite{dwork:epsilondelta} to have two parameters $\epsilon$ and $\delta$ playing multiplicative and additive roles in the likelihood constraint.  We refer to the original DP as $\epsilon$-DP and the relaxed form as $(\epsilon,\delta)$-DP.  In this notation, $\epsilon$-DP is simply $(\epsilon,0)$-DP.

The claim herein is that MI-DP is sandwiched between these two definitions in the following sense:  It is weaker than $\epsilon$-DP but stronger than $(\epsilon,\delta)$-DP.  That is, a mechanism that satisfies $\epsilon$-DP also satisfies $\epsilon$-MI-DP.\footnote{The other direction is not true in general, as there exist mechanisms which satisfy $\epsilon$-MI-DP but not $\epsilon'$-DP for any $\epsilon'$.}  Similarly, if $\epsilon$-MI-DP holds then $(\epsilon',\delta)$-DP also must hold, where $\epsilon'$ and $\delta$ vanish as $\epsilon$ goes to zero.  In fact, the connection between MI-DP and $(\epsilon,\delta)$-DP is an equivalence if either the domain or range of the query mechanism is a finite set.

The advantage of this alternative but equally strong definition of differential privacy is that mutual information is a well-understood quantity.  It provides a clear picture of what differential privacy does and does not guarantee.  Furthermore, several properties of differential privacy are immediate to prove in this form. 

While the mathematics of differential privacy, in its standard form, are straightforward, an intuitive understanding can be elusive.  The definition of $(\epsilon,\delta)$-DP involves a notion of neighboring database instances.  Upon examination one realizes that this has the affect of assuming that the adversary has already learned about all but one entry in the database and is only trying to gather additional information about the remaining entry.  We refer to this as the strong adversary assumption, which is implicit in the definition of differential privacy.  Notice that MI-DP needs no definition of neighborhood.  The strong adversary assumption is made explicit in the conditioning within the conditional mutual information term.

The strong adversary assumption is both a feature and a vulnerability of the definition of differential privacy.  It is a feature when recruiting individual participants for a survey.  The individual can decide whether or not to participate but cannot do anything about the information contributed by others (which may inform on them indirectly).  Differential privacy assures them that even with access to everyone else's responses, the survey reports will not further reveal anything about their individual response.  However, DP also has its shortcomings as a privacy guarantee.  Among all adversaries with different prior knowledge of the database, the strong adversary may not be the one which benefits the most from the query output.  Indeed, it is shown in \cite{kifer:nfl} that a weaker adversary can compromise privacy severely if the entries in the database are correlated, which is quite typical in certain applications such as social networks.

As an equivalent privacy metric, MI-DP benefits and suffers in the same way.  Fortunately, the definition of MI-DP puts this potential weakness in plain sight.  It shows explicitly that information leakage is only being restricted conditioned on the remainder of the database being known.  Clearly, this does not bound the unconditional mutual information when correlations are present.

Somewhat paradoxically, mutual information can serve simultaneously as both a measure of privacy, as in MI-DP, and as a quantification of utility---for example, the mutual information between the entire database and the query response, $I(X^n;Y)$.  The close connection between mutual information and estimation and detection further captures the privacy-utility trade-off.

Several works \cite{mcgregor:two, de:lower, barthe:min, alvim:min, duchi:lp} relate mutual information to differential privacy by upper bounding mutual information given a differential privacy achieving mechanism.  One common point of these works is that they all use unconditional mutual information rather than conditional.  Often, the conclusion is a bound on the unconditional mutual information between the whole database and the private output.  The use of conditional mutual information in Definition~\ref{def:midp} captures the prior knowledge of the database possessed by a potential adversary (i.e. the strong adversary assumption implicit in DP).  This is crucial in developing an equivalence with the standard DP definition.

Another crucial ingredient that some of the literature fails to properly incorporate (e.g. \cite{wang-ying-zhang14}) when applying mutual information to differential privacy is that differential privacy is a property of the query mechanism and assumes no specific prior distribution on the database.  Mutual information, on the other hand, is not well defined without a joint distribution, which must include a distribution for the database.  The remedy is to maximize the mutual information over all possible distributions on the database, as seen in the definition of MI-DP in Definition~\ref{def:midp}.\footnote{This approach, using worst-case database distribution, appears in various works throughout the literature, e.g. \cite{duchi-jordan-wainwright14}.}  Had we defined MI-DP with respect to any particular database distribution (e.g. with independent and identically distributed entries), we would have significantly reduced its strength as a privacy metric.  The formula in Definition~\ref{def:midp} in fact looks like a channel capacity formula one would encounter in expressing the fundamental limit of communication through a noisy channel.  This maximization removes any distributional assumption and makes MI-DP a property of the mechanism itself.

\section{Preliminaries}
\label{sec:pre}

\subsection{Notation}

The set $\{1, 2, \cdots, m\}$ is denoted as $[m]$. An index set ${\cal I}$ is a subset of $[n]$ whose elements are enumerated as $(i_1, \cdots, i_{|{\cal I}|})$, where $|\cdot|$ denotes the cardinality of a set.

We use $X^n$ as shorthand notation for the sequence of random variables $(X_1, \cdots, X_n)$. The symbol $\xmi$ denotes the sequence of $n - 1$ random variables $(X_1, \cdots, X_{i - 1}, X_{i + 1}, \cdots, X_n)$, in other words, all of $X^n$ except $X_i$. The lower case symbol $x^{-i}$ is an instance of $\xmi$. For any index set ${\cal I}$, we use $X_{\cal I}$ to denote the sequence of random variables $(X_{i_1}, \cdots, X_{i_{|{\cal I}|}})$ specified by ${\cal I}$. Similarly the lower case $x_{\cal I} = (x_{i_1}, \cdots, x_{i_{|{\cal I}|}})$ is an instance of $X_{\cal I}$.

A database $X^n$ consists of $n$ entries, where the $i$-th entry takes values from ${\cal X}_i$. 
\begin{definition}[Neighbor]
\label{def:neighbor}
    Two database instances $x^n$ and $\tilde{x}^n$ are {\em neighbors} if they differ in only one entry.  In other words,
    \begin{equation}
        d_H (x^n, \tilde{x}^n) = 1,
    \end{equation}
    where $d_H(\cdot, \cdot)$ is Hamming distance.
\end{definition}

The output of a privacy mechanism is a random variable represented as $Y$ and takes values from $\cy$.

\subsection{Statistical Indistinguishability}

Two probability distributions can be considered statistically indistinguishable if they are close under an appropriate metric.  The criterion for indistinguishability used in the standard definition of differential privacy is the following.

\begin{definition}[$(\epsilon, \delta)$-Closeness]
\label{def:ed close}
    Two probability distributions $P$ and $Q$ over the same measurable space $(\Omega, {\cal F})$ are $(\epsilon, \delta)$-close, denoted as
    \begin{equation}
        P \stackrel{(\epsilon, \delta)}{\approx} Q
    \end{equation}
    if
    \begin{align}
        P(A) &\leq e^{\epsilon} Q(A) + \delta, \quad \forall A \in {\cal F}, \label{eq:ed p less q} \\
        Q(A) &\leq e^{\epsilon} P(A) + \delta, \quad \forall A \in {\cal F}. \label{eq:ed q less p}
    \end{align}
\end{definition}

Consider two special cases,  $\delta=0$ and $\epsilon=0$.  If $\delta = 0$, then $P$ and $Q$ are mutually absolutely continuous, denoted as $P \ll \gg Q$, and $(\epsilon, 0)$-closeness is a statement about the Radon-Nikodym derivative $\frac{d P}{d Q}$:
\begin{equation}
    \label{eq:e0 restatement}
    P \stackrel{(\epsilon, 0)}{\approx} Q \quad \Longleftrightarrow \quad \left| \ln \frac{d P}{d Q}(a) \right| \leq \epsilon \quad \forall a \in \Omega.
\end{equation}
On the other hand, with $\epsilon=0$, $(0, \delta)$-closeness is a statement about the total variation distance:
\begin{equation}
    \label{eq:0d restatement}
    P \stackrel{(0, \delta)}{\approx} Q \quad \Longleftrightarrow \quad \|P - Q \|_{TV} \leq \delta.
\end{equation}

We can also relate $(\epsilon, \delta)$-closeness to Kullback-Leibler divergence, denoted as $D(\cdot \| \cdot)$.  For example, by relaxing the right side of \eqref{eq:e0 restatement} to be an expected value rather than a statement about all $a \in \Omega$, we immediately get the following implication:
\begin{equation}
    P \stackrel{(\epsilon, 0)}{\approx} Q \quad \Longrightarrow \quad
    \begin{aligned}
    D(P\|Q) &\leq \epsilon \text{ nats}, \\
    D(Q\|P) &\leq \epsilon \text{ nats}.
    \end{aligned}
\end{equation}
Tighter expressions of the relationship to Kullback-Leibler divergence are given next, in Properties~\ref{PROPERTY:E TO KL} and \ref{property:kl to d}.  We give a proof of Property~\ref{PROPERTY:E TO KL} in Appendix~\ref{PROOF:E TO KL}.

\begin{property}
\label{PROPERTY:E TO KL}
    \begin{equation}
        P \stackrel{(\epsilon, 0)}{\approx} Q \quad \Longrightarrow \quad
        \begin{aligned}
            D(P\|Q) &\leq \min \left\{ \epsilon, \epsilon^2 \right\} \text{ nats}, \\
            D(Q\|P) &\leq \min \left\{ \epsilon, \epsilon^2 \right\} \text{ nats}.
        \end{aligned}
    \end{equation}
    In fact, the tightest possible statement of this form is
    \begin{equation}
    \label{eq:e to kl tightest}
        P \stackrel{(\epsilon, 0)}{\approx} Q \quad \Longrightarrow \quad
        \begin{aligned}
            D(P\|Q) &\leq \epsilon \frac{\left( e^{\epsilon} - 1 \right) \left( 1 - e^{-\epsilon} \right)}{\left( e^{\epsilon} - 1 \right) + \left( 1 - e^{-\epsilon} \right)} \text{ nats}, \\
            D(Q\|P) &\leq \epsilon \frac{\left( e^{\epsilon} - 1 \right) \left( 1 - e^{-\epsilon} \right)}{\left( e^{\epsilon} - 1 \right) + \left( 1 - e^{-\epsilon} \right)} \text{ nats}.
        \end{aligned}
    \end{equation}
    Equality on the right side of \eqref{eq:e to kl tightest} can be achieved with binary distributions.  For small $\epsilon$, the right side of \eqref{eq:e to kl tightest} is asymptotically $\frac{1}{2} \epsilon^2$ nats.
\end{property}

\begin{property}
\label{property:kl to d}
    By Pinsker's inequality,
    \begin{equation}
        D(P\|Q) \leq \epsilon \text{ nats} \quad \Longrightarrow \quad P \stackrel{\left( 0, \sqrt{\epsilon/2} \right)}{\approx} Q.
    \end{equation}
\end{property}

Property~\ref{PROPERTY:E TO KL} and Property~\ref{property:kl to d} are strict in the sense that the reverse implications are not true in any form (i.e. closeness bounds on the right, no matter how small the parameters, do not even imply finiteness of the parameters on the left).  Also, we already mentioned that Property~\ref{PROPERTY:E TO KL} is tight.  Property~\ref{property:kl to d} is known to be tight up to a multiplicative constant.

The quantities arising in the above definitions and properties have concrete connections to inference.  Total variation distance precisely captures the error probability in a binary hypothesis test.  That is, one minus the total variation distance is the minimum sum of the two types of binary error probability.  Kullback-Leibler divergence precisely captures the asymptotic hypothesis testing error upon observing many independent observations \cite[Chapter~11]{cover:it}.  The strongest of these metrics, $(\epsilon, 0)$-closeness, has an interpretation in the Bayesian setting as a bound on the Bayes factor.  That is, the log-posterior-odds-ratio cannot change by more than $\epsilon$ due to the observation.  Finally, $(\epsilon,\delta)$-closeness is shown in \cite{kairouz:composition} to be precisely a piecewise linear constraint on the error region in a binary hypothesis test.  By inspection of that relationship, the following property is apparent (proven in Appendix~\ref{PROOF:E TO D}).
\begin{property}
\label{PROPERTY:E TO D}
    For any non-negative $\epsilon'<\epsilon$, let $\delta' = 1 - \frac{\left( e^{\epsilon'} + 1 \right)(1-\delta)}{e^{\epsilon}+1}$.
        \begin{equation}
        P \stackrel{\left( \epsilon, \delta \right)}{\approx} Q \quad \Longrightarrow \quad P \stackrel{\left( \epsilon', \delta'\right)}{\approx} Q.
    \end{equation}
\end{property}

Property~\ref{PROPERTY:E TO D} is the tightest possible trade-off between $\epsilon$ and $\delta$ with respect to $(\epsilon, \delta)$-closeness.  Notice that $\delta'>\delta$.  It is not possible for a larger $\delta$ to imply a smaller one, for any finite $\epsilon$ and $\epsilon'$.

\subsection{Differential Privacy}
\label{sec:dp}

The definition of $(\epsilon,\delta)$-DP in Definition~\ref{def:eddp} has been now made precise with Definition~\ref{def:neighbor} (neighbor) and Definition~\ref{def:ed close} ($(\epsilon, \delta)$-closeness).

We define $\epsilon$-DP and $(\delta)$-DP by setting either of the two parameters to zero.

\begin{definition}[$\epsilon$-Differential Privacy \cite{dwork:dp}]
A randomized mechanism $P_{Y|X^n}$ satisfies $\epsilon$-DP if it satisfies $(\epsilon, 0)$-DP.
\end{definition}

\begin{definition}[$(\delta)$-Differential Privacy]
A randomized mechanism $P_{Y|X^n}$ satisfies $(\delta)$-DP if it satisfies $(0, \delta)$-DP.
\end{definition}

Mutual-information differential privacy was defined in the introduction in Definition~\ref{def:midp}.

Finally, let us define one additional privacy metric based on Kullback-Leibler divergence, which we will call KL-DP.

\begin{definition}[KL Differential Privacy]
\label{def:kldp}
    A randomized mechanism $P_{Y|X^n}$ satisfies $\epsilon$-KL-DP if for all neighboring database instances $x^n$ and $\tilde{x}^n$
    \begin{equation}
        D \left( P_{Y|X^n = x^n} \middle\| P_{Y|X^n = \tilde{x}^n} \right) \leq \epsilon \text{ nats}.
    \end{equation}
\end{definition}

\subsection{Ordering of Privacy Metrics}

This work is about showing equivalence of privacy metrics.  In order to do so, we must define an ordering.

\begin{definition}[Stronger Privacy Metric]
    As a place-holder, take $\alpha$-DP and $\beta$-DP to represent two generic privacy guarantees with positive parameters $\alpha$ and $\beta$.  We say that $\alpha$-DP is stronger than $\beta$-DP, denoted as
    \begin{equation}
        \text{$\alpha$-DP} \succeq \text{$\beta$-DP},
    \end{equation}
    if for all $\beta'>0$ there exists an $\alpha'>0$ such that
    \begin{equation}
        \text{$\alpha$'-DP} \quad \Longrightarrow \quad \text{$\beta$'-DP}.
    \end{equation}
    If the parameters are vectors, then $\beta'>0$ and $\alpha'>0$ should be interpreted as inequalities on each coordinate.
\end{definition}

\begin{example}
\label{example:ordering}
    It is clear that $\epsilon$-DP $\succeq$ $(\epsilon,\delta)$-DP and $(\delta)$-DP $\succeq$ $(\epsilon,\delta)$-DP, since $\epsilon$-DP implies $(\epsilon,\delta)$-DP for any non-negative $\delta$, and likewise for $(\delta)$-DP, by definition.
    
    Also, $(\epsilon, \delta)$-DP $=$ $(\delta)$-DP by Property~\ref{PROPERTY:E TO D}.  Notice that even if we set $\epsilon'=0$, the quantity $\delta'$, as defined in the property, goes to zero as $\epsilon$ and $\delta$ go to zero.
\end{example}

\section{Main Result}
\label{section main result}

\subsection{Equivalence}

The emphasis of this work is the equivalence of mutual-information differential privacy with classical differential privacy.

\begin{thm}[Main Result]
\label{thm:main}
    \begin{equation}
    \label{eq:sandwich}
        \text{$\epsilon$-DP} \succeq \text{MI-DP} \succeq \text{$(\epsilon, \delta)$-DP}.
    \end{equation}
    Furthermore, if the cardinality of the database entries or the query response is bounded, then
    \begin{equation}
    \label{eq:mi equals ed}
        \text{MI-DP} = \text{$(\epsilon,\delta)$-DP},
    \end{equation}
    where the relationship $(\epsilon,\delta)$-DP $\succeq$ MI-DP is dependent on the cardinality bound
    \begin{equation}
        \min \left\{ |{\cal Y}|, \max_i |{\cal X}_i| \right\}.
    \end{equation}
\end{thm}

Precise bounds for the privacy parameters are given in the three lemmas in Section~\ref{section proof main}.

\subsection{Related Work}

Using information theoretic measures to quantify the privacy guarantee of differential privacy is not a new idea. An upper bound of mutual information is shown in \cite{mcgregor:two} in a two-party differential privacy setting. Later this upper bound is used in \cite{de:lower} to get $I(X^n;Y) \leq 3 \epsilon n$. In \cite{alvim:min} and \cite{barthe:min}, min-entropy is considered rather than the usual Shannon entropy, and upper bounds are proven.  In fact, \cite[Corollary~1]{barthe:min} implies an ordering relationship similar to the first inequality of \eqref{eq:sandwich} but for min-entropy based information leakage with only a single database entry.  In \cite{wang-ying-zhang14}, a ``mutual information privacy'' metric is defined and studied.

These works have in common that they all consider the use of unconditional mutual information. This doesn't capture the structure in the definition of differential privacy and the bounds are limited to the mutual information between the whole database and the sanitized query output, with no focus on individual entries.  Needless to say, an equivalence is not established.

Some of the information theory literature bares resemblance to this work.  In \cite{bellare:semantic}, similar proof steps to this work are used to show an equivalence between semantic security and a mutual information constraint.  As in this work, there is a maximization over distributions of inputs to the randomized mechanism; however, conditional mutual information is not a part of that result, while it is a necessary ingredient here.  Also, the notion of $(\epsilon, \delta)$-closeness goes by the name of $E_{\gamma}$ distance in some of the information theory literature, such as \cite{polyanskiy2010channel} and \cite{liu2015resolvability}.  Specifically, $P \stackrel{(\epsilon, \delta)}{\approx} Q$ is equivalent to the pair of statements $E_{e^{\epsilon}} (P\|Q) \leq \delta$ and $E_{e^{\epsilon}} (Q\|P) \leq \delta$.

\subsection{Proof of Theorem~\ref{thm:main}}
\label{section proof main}

We prove \eqref{eq:sandwich} of Theorem~\ref{thm:main} by proving a stronger chain of inequalities:
\begin{equation}
    \text{$\epsilon$-DP} \stackrel{\text{(A)}}{\succeq} \text{KL-DP} \stackrel{\text{(B)}}{\succeq} \text{MI-DP} \stackrel{\text{(C)}}{\succeq} \text{$(\delta)$-DP} \stackrel{\text{(D)}}{=} \text{$(\epsilon, \delta)$-DP}.
\end{equation}

It is worth noting both (A) and (B) are in fact strict orderings ($\succ$)---the reverse implications do not hold, even if cardinality bounds are assumed.

We now state the components of the proof in separate lemmas.  Orderings (A) and (B) are the subject of Lemma~\ref{lem:e implies mi}, and ordering (C) is handled by Lemma~\ref{LEM:MI IMPLIES D}.  Equality (D) comes from Property~\ref{PROPERTY:E TO D}, as discussed in Example~\ref{example:ordering}.

\begin{lem}[Orderings (A) and (B)]
\label{lem:e implies mi}
    \begin{align}
        \text{$\epsilon$-DP} \quad &\Longrightarrow \quad \text{$\min \left\{\epsilon, \epsilon^2 \right\}$-KL-DP}, \label{eq:e implies kl} \\
        \text{$\epsilon$-KL-DP} \quad &\Longrightarrow \quad \text{$\epsilon$-MI-DP}. \label{eq:kl implies mi}
    \end{align}
    Therefore,
    \begin{equation}
        \text{$\epsilon$-DP} \quad \Longrightarrow \quad \text{$\min \left\{\epsilon, \epsilon^2 \right\}$-MI-DP}. \label{eq:e implies mi}
    \end{equation}
\end{lem}

\begin{proof}
    The first statement, \eqref{eq:e implies kl}, is established by Property~\ref{PROPERTY:E TO KL}.  Both $\epsilon$-DP and KL-DP are defined the same way in terms of neighboring database instances.
    
    The second statement, \eqref{eq:kl implies mi}, is best understood through the geometric interpretation of capacity as the radius of the information ball \cite[Theorem~13.1.1]{cover:it}.  The radius cannot be more than the maximum of pairwise distances.  However, we will not directly use that machinery here.  Instead, consider the following direct proof.
    
    Start by assuming that the randomized mechanism $P_{Y|X^n}$ satisfies $\epsilon$-KL-DP.  Let $i \in \{1,\ldots,n\}$ and $P_{X^n}$ be arbitrary.  For notational clarity, let $\bar{X}^n \sim P_{X^n}$, and begin with a representation of conditional mutual information for a general distribution in terms of Kullback-Leibler divergence:
    \begin{equation}
        I(X_i;Y|X^{-i}) = \mathbb{E} \left[ D \left( P_{Y|X^n = \bar{X}^n} \middle\| P_{Y|X^{-i}=\bar{X}^{-i}} \right) \right]
    \end{equation}
    
    Now we bound $D \left( P_{Y|X^n=x^n} \middle\| P_{Y|X^{-i}=x^{-i}} \right)$ for each instance $x^n$.  Fix $x^n$ arbitrarily, and let $\tilde{X} \sim P_{X_i|X^{-i}=x^{-i}}$.  Consider,
    \begin{equation}
        P_{Y|X^{-i}=x^{-i}} = \mathbb{E} \left[ P_{Y|X_i = \tilde{X},X^{-i}=x^{-i}} \right].
    \end{equation}
    
    Therefore, by Jensen's inequality, and using the fact that $D(\cdot \| \cdot)$ is convex in the second argument, we conclude,
    \begin{align}
        &D \left( P_{Y|X^n = x^n} \middle\| P_{Y|X^{-i}=x^{-i}} \right) \nonumber \\
        &= D \left( P_{Y|X^n = x^n} \middle\| \mathbb{E} \left[ P_{Y|X_i=\tilde{X},X^{-i}=x^{-i}} \right] \right) \nonumber \\
        &\leq \mathbb{E} \left[ D \left( P_{Y|X^n = x^n} \middle\| P_{Y|X_i=\tilde{X},X^{-i}=x^{-i}} \right) \right] \nonumber \\
        &\leq \epsilon \text{ nats},
    \end{align}
    where the last inequality is due to the fact that any two databases that agree on $X^{-i}$ are neighbors.
\end{proof}

\begin{lem}[Ordering (C)]
\label{LEM:MI IMPLIES D}
    \begin{equation}
        \text{$\epsilon$-MI-DP} \quad \Longrightarrow \quad \text{$\left( 0, \sqrt{2 \epsilon} \right)$-DP}. \label{eq:mi to d simple}
    \end{equation}
    In fact, the tightest possible statement of this form is
    \begin{equation}
        \text{$\epsilon$-MI-DP} \quad \Longrightarrow \quad \text{$\left( 0, \delta' \right)$-DP}, \label{eq:mi to d tight}
    \end{equation}
    with $\delta' = 1 - 2 h^{-1} (\ln 2 - \epsilon)$, where $h^{-1}$ is the inverse of the increasing part of the binary entropy function in units of nats.  This formula holds for $\epsilon \in [0, \ln 2]$.  For $\epsilon > \ln 2$, the implication becomes $(1)$-DP, which is vacuous.
    
    The claim in \eqref{eq:mi to d simple} is looser than that in \eqref{eq:mi to d tight} but asymptotically tight for small $\epsilon$.
\end{lem}

\begin{proof}
    The essence of this claim is found in the binary case, with a binary database and a binary query response.  We show this reduction first.
    
    Start by assuming that the randomized mechanism $P_{Y|X^n}$ satisfies $\epsilon$-MI-DP.  Consider an arbitrary pair of neighboring database instances $x^n$ and $\tilde{x}^n$, and let $i$ be the location where they differ.  Denote by $\Delta_{x^n,\tilde{x}^n}$ the subset of probability distributions over the space of databases ${\cal D}$ that only put positive mass on $x^n$ and $\tilde{x}^n$.  Therefore, all distributions in $\Delta_{x^n,\tilde{x}^n}$ are binary, and $X^{-i}$ is deterministic with respect to any of them.

    Also, let $A$ be an arbitrary measurable subset of ${\cal Y}$.  Consider the indicator function
    \begin{equation}
        B(y) =
        \begin{cases}
            1, & y \in A, \\
            0, & y \notin A. \\
        \end{cases}
    \end{equation}
    The random variable $B$ is the binary function $B(Y)$.
    
    \begin{align}
        \max_{P_{X^n} \in \Delta_{x^n, \tilde{x}^n}} I (X_i; B) &\stackrel{(a)}{\leq} \max_{P_{X^n} \in \Delta_{x^n, \tilde{x}^n}} I (X_i; Y) \nonumber \\
        &\stackrel{(b)}{=} \max_{P_{X^n} \in \Delta_{x^n, \tilde{x}^n}} I (X_i; Y | X^{-i}) \nonumber \\
        &\leq \sup_{P_{X^n}} I (X_i; Y | X^{-i}) \nonumber \\
        &\stackrel{(c)}{\leq} \epsilon \text{ nats}, \label{eq:reduction to binary}
    \end{align}
    where (a) is due to the data processing inequality, (b) comes from the fact that $X^{-i}$ is deterministic for all distributions in $\Delta_{x^n, \tilde{x}^n}$, and (c) is by assumption of $\epsilon$-MI-DP.
    
    To summarize, we have arrived at a binary input and binary output randomized mechanism $P_{B|X_i}$, where $X_i \in \{ x_i, \tilde{x}_i \}$, defined by
    \begin{align}
        P_{B|X_i = x_i}(\{1\}) &= P_{Y|X^n = x^n}(A), \\
        P_{B|X_i = \tilde{x}_i}(\{1\}) &= P_{Y|X^n = \tilde{x}^n}(A).
    \end{align}
    This mechanism is shown in \eqref{eq:reduction to binary} to satisfy $\epsilon$-MI-DP.  Also, since $A$, $x^n$, and $\tilde{x}^n$ were chosen arbitrarily, any $(\delta)$-DP claim that can be made about $P_{B|X_i}$ must also hold for $P_{Y|X^n}$.
    
    In Appendix~\ref{PROOF:MI IMPLIES D}, we complete the proof by showing that Lemma~\ref{LEM:MI IMPLIES D} holds for all randomized mechanisms with a binary input and binary output, and that the characterization is tight.
    
    A more complete characterization is also possible, of the form
    \begin{equation}
        \text{$\epsilon$-MI-DP} \quad \Longrightarrow \quad \text{$\left( \epsilon', \delta' \right)$-DP}, \label{eq:mi to ed}
    \end{equation}
    for a particular set of values $(\epsilon', \delta')$ which, among other things, has the property that $\delta'$ must be greater than some positive threshold which depends on $\epsilon$, and as $\delta'$ approaches this threshold, $\epsilon'$ must go to infinity.  This characterization is also arrived at by first reducing to the binary case as we have done above.  However, a description of the trade-off is too unwieldy for this discussion.
\end{proof}

We prove \eqref{eq:mi equals ed} of Theorem~\ref{thm:main} with the following claim.  The proof is given in Appendix~\ref{PROOF:D TO MI}.

\begin{lem}[Reverse direction]
\label{LEM:D IMPLIES MI}
    If $|{\cal X}_i|$ is finite for all $i \in \{X_1,\dots,X_n\}$, or if $|{\cal Y}|$ is finite, then
    \begin{equation}
        \text{$(\delta)$-DP} \quad \Longrightarrow \quad \text{$\epsilon'$-MI-DP},
    \end{equation}
    where, for any $\delta \in [0,1]$,
    \begin{equation}
        \epsilon' = 2 h(\delta) + 2 \delta \ln \left( \min \left\{ |{\cal Y}|, \max_i |{\cal X}_i| + 1 \right\} \right).
    \end{equation}
\end{lem}

Slightly tighter bounds can be found in \eqref{eq:d to mi tighter y} and \eqref{eq:d to mi tighter x} of the proof.  Although these bounds may have some looseness, the following example shows that they get roughly within a factor of two of the correct scaling for large cardinalities.

\begin{example}[Erasure channel]
\label{eg:delta}
Consider a database with only one entry, $X_1$.  Let ${\cal X}_1 = [N]$ and ${\cal Y} = [N] \cup \{0\}$.  Define
\begin{equation}
    P_{Y|X_1=x_1} = 
    \begin{cases}
        1 - \delta, & y = 0, \\
        \delta, & y = x_1, \\
        0, & \text{otherwise}.
    \end{cases}
\end{equation}

This randomized mechanism is usually referred to as an erasure channel, where the output $Y=0$ is considered an erasure.  It is known that the capacity of this channel is
\begin{equation}
    C = \delta \log N,
\end{equation}
where $N = |{\cal X}_1| = |{\cal Y}|-1$.  This implies that there exists a distribution of the database (in this case, the uniform distribution for $X_1 \in {\cal X}_1$) such that
\begin{align}
    I(X_1;Y|X^{-1}) &= I(X_1;Y) \nonumber \\
    &= \delta \log |{\cal X}_1| \nonumber \\
    &= \delta \log (|{\cal Y}| - 1).
\end{align}
\end{example}

\section{Properties of Diff. Privacy}
\label{sec:midp}

Now that we have MI-DP as an equivalent metric of privacy, we explore the insights that this brings and simple proofs of properties about privacy.

The following are three basic and well-known properties of mutual information:

\begin{property}
\label{property:mi order independent}
    If $U$ is independent of $W$, then
    \begin{equation}
        I(U;V|W) \geq I(U;V).
    \end{equation}
\end{property}

\begin{property}
\label{property:mi order markov}
    If $U$, $V$, and $W$ form a Markov chain $U - V - W$, meaning that $U$ and $W$ are conditionally independent given $V$, then
    \begin{equation}
        I(U;V|W) \leq I(U;V).
    \end{equation}
\end{property}

\begin{property}[Chain rule]
\label{property:mi chain rule}
    \begin{equation}
        I(U;V,W) = I(U;V) + I(U;W|V).
    \end{equation}
\end{property}

We will use these three properties (sometimes conditioned on other random variables) to make claims about MI-DP.

\subsection{The Strong Adversary Assumption}
\label{section:strong adversary}

We refer to a strong adversary as one who knows the entire database except for any one entry $X_i$.  Differential privacy is implicitly designed as a protection against further information leakage to this adversary.  The definition of MI-DP, now shown to be equivalent, makes this attribute explicit by conditioning on the remainder of the database and bounding $I(X_i;Y|X^{-i})$.  But how much information does the sanitized output $Y$ leak to an adversary with no prior knowledge?

In \cite{kifer:nfl}, this is referred to as {\em evidence of participation}.  In the mutual information context, this may be measured by the unconditional mutual information $I(X_i;Y)$.  It is pointed out in \cite{kifer:nfl} that if the entries of the database are independent, the evidence of participation can be protected properly by differential privacy.  This claim is straightforward using MI-DP in light of Property~\ref{property:mi order independent}.

\begin{cor}[Independent Data]
If $\{X_i\}_{i = 1}^n$ are mutually independent and $P_{Y|X^n}$ satisfies $\epsilon$-MI-DP, then
\begin{equation}
    \sup_{i,P_{X_i}P_\xmi} I(X_i;Y) \leq \sup_{i,P_{X_i}P_\xmi} I(X_i;Y|\xmi) \leq \epsilon \text{ nats}.
\end{equation}
\end{cor}

On the other hand, it is often the case that entries of a database are correlated.  Differential privacy does not provide a strong guarantee about the evidence of participation in general.  Consider the following familiar example:

\begin{example}[Correlated database]
\label{eg:dep}
Consider a database with $n$ binary entries.  A data curator decides to release the mean of all entries and chooses the Laplace mechanism.  Noise with distribution $\rm{Lap}(\frac{1}{n\epsilon})$ is added to the sample mean to ensure $\epsilon$-DP (also $\epsilon$-MI-DP by Lemma~\ref{lem:e implies mi}).

Now suppose all database entries are in fact equal to each other (maximally correlated).  Let $X \sim \rm{Bern}(0.5)$ and $X_i = X$ for all $i \in \{1,\ldots,n\}$.  For large enough $n$, the noise added is negligible, and the binary value of the sample mean can be estimated with high accuracy, revealing each individual entry.  In terms of mutual information, $I(X_i;Y) \approx 1$ bit for each $i$ even while $I(X_i;Y|X^{-i}) = 0$ because $H(X_i|X^{-i})=0$.
\end{example}

\subsection{Composition}

Among the most important properties of differential privacy is composability.  This states that a collection of queries, each satisfying differential privacy, collectively satisfies differential privacy with a parameter scaled proportional to the number of queries.

A great deal of effort has been made in deriving tight composition theorems for differential privacy.  A straight-forward composition theorem can be found in \cite{dwork:epsilondelta}.  More intricate trade-offs can be found in \cite{dwork-rothblum-vadhan10} and \cite{kairouz:composition}, with the latter establishing a tight characterization.

The following claims for MI-DP mirror those found in \cite{mcsherry:comp} for $(\epsilon,0)$-DP and are in fact tight.

\begin{cor}[Conditionally independent queries]
\label{cor:conditionally independent}
    If several query responses $\{Y_1,\ldots,Y_k\}$ are produced conditionally independently given the database, and each mechanism $P_{Y_j|X^n}$ satisfies $\epsilon_j$-MI-DP individually, then as a collection $P_{Y^k|X^n}$ satisfies $\left( \sum_j \epsilon_j \right)$-MI-DP.
\end{cor}

\begin{proof}
    For any $i$ and $P_{X^n}$, the chain rule of mutual information (Property~\ref{property:mi chain rule}) gives (a), and Property~\ref{property:mi order markov} gives (b):
    \begin{align}
        I(X_i;Y^k|X^{-1}) &\stackrel{(a)}{=}  \sum_{j=1}^k I(X_i;Y_j|X^{-i},Y^{j-1}) \nonumber \\
        &\stackrel{(b)}{\leq} \sum_{j=1}^k I(X_i;Y_j|X^{-i}) \label{eq:mi composition line to skip} \nonumber \\
        &\leq \sum_{j=1}^k \epsilon_j \text{ nats}.
    \end{align}
\end{proof}

Corollary~\ref{cor:conditionally independent} states that the effect of releasing multiple conditionally independent query responses has no more than an additive effect on the parameter of privacy.  It is worth noting two important points.  First, query responses that are not conditionally independent (i.e. the noise from one query response is somehow reused in the next) have no such guarantee, as the following example illustrates.

\begin{example}[Correlated query responses]
\label{example:correlated responses}
    Consider a database where each entry has a finite alphabet $|{\cal X}_i| \leq \infty$.  Consider two outputs of a query mechanism, $Y_1 = X_1 \oplus U$ and $Y_2 = U$, where $U$ is a uniformly distributed random variable on the set $\{1,\ldots,|{\cal X}_1|\}$, independent of the database instance, and $\oplus$ is addition modulo $|{\cal X}_1|$.  In other words, the first output $Y_1$ is $X_1$ encrypted by a one-time pad, and the second output $Y_2$ is the key to the one-time pad.  Clearly, the combination of $Y_1$ and $Y_2$ reveals $X_1$ and violates differential privacy.
\end{example}

On the other hand, Example~\ref{example:correlated responses} does not imply that correlated query responses should not be considered.  Quite to the contrary, query responses that are carefully constructed to be correlated with each other have the potential to achieve significantly better privacy after multiple queries, as demonstrated in \cite{hardt-rothblum10} and \cite{blum-avrim-ligett13}.

In general, the same composition claim of Corollary~\ref{cor:conditionally independent} holds even if the query responses are correlated as long as each response in sequence is specifically designed to satisfy differential privacy even with respect to the previous responses.  The following corollary states this claim, and the proof follows directly from the proof of Corollary~\ref{cor:conditionally independent} simply by skipping \eqref{eq:mi composition line to skip}.

\begin{cor}[Sequential queries]
    If several query responses $\{Y_1,\ldots,Y_k\}$ are produced in sequence, and each mechanism $P_{Y_j|X^n,Y^{j-1}}$ satisfies $\epsilon_j$-MI-DP individually, then as a collection $P_{Y^k|X^n}$ satisfies $\left( \sum_j \epsilon_j \right)$-MI-DP.
\end{cor}

The next claim is about query responses that each depend on different subsets of the database.

\begin{cor}[Partial queries]
    If several query responses $\{Y_1,\ldots,Y_k\}$ are produced conditionally independently of each other from disjoint subsets of the database entries, denoted as $X_{{\cal I}_1},\ldots,X_{{\cal I}_k}$, with each mechanism $P_{Y_j|X_{{\cal I}_j}}$ satisfying $\epsilon$-MI-DP individually, then as a collection $P_{Y^k|X^n}$ also satisfies $\epsilon$-MI-DP.
\end{cor}

\begin{proof}
    Let $f(i)$ be the index $j$ such that $i \in {\cal I}_j$.  For any $i$ and $P_{X^n}$, the chain rule of mutual information (Property~\ref{property:mi chain rule}) gives:
    \begin{align}
        I(X_i;Y^k|X^{-1}) &=  I(X_i;Y_{f(i)}|X^{-i}) + I(X_i;Y^{-f(i)}|X^{-i},Y_{f(i)}) \nonumber \\
        &= I(X_i;Y_{f(i)}|X^{-i}) \nonumber \\
        &\leq \epsilon \text{ nats}.
    \end{align}
\end{proof}

\section{A Discrepancy}

While most properties of $\epsilon$-DP or $(\epsilon, \delta)$-DP are also properties of MI-DP, it turns out that one basic property does not carry over.

Differential privacy is defined with respect to neighboring database instances.  What privacy can be guaranteed if some bounded number of entries are changed in the database?  Similar to the composition properties, the closeness of the output distribution scales proportionally with the number of database changes.  The following properties are obtained by repeated application \eqref{eq:ed p less q} and \eqref{eq:ed q less p} from Definition~\ref{def:ed close}.

\begin{property}[Epsilon]
\label{property:group e}
    Suppose $x^n$ and $\tilde{x}^n$ are instances of the database that differ in at most $k$ entries, and that the randomized mechanism $P_{Y|X^n}$ is $\epsilon$-DP.  Then
    \begin{equation}
        P_{Y|X^n=x^n} \stackrel{(k \epsilon, 0)}{\approx} P_{Y|X^n=\tilde{x}^n}.
    \end{equation}
\end{property}

\begin{property}[Delta]
\label{property:group d}
    Suppose $x^n$ and $\tilde{x}^n$ are instances of the database that differ in at most $k$ entries, and that the randomized mechanism $P_{Y|X^n}$ is $(\delta)$-DP.  Then
    \begin{equation}
        P_{Y|X^n=x^n} \stackrel{(0, k \delta)}{\approx} P_{Y|X^n=\tilde{x}^n}.
    \end{equation}
\end{property}

\begin{property}[General]
\label{property:group both}
    Suppose $x^n$ and $\tilde{x}^n$ are instances of the database that differ in at most $k$ entries, and that the randomized mechanism $P_{Y|X^n}$ is $(\epsilon, \delta)$-DP.  Then
    \begin{equation}
        P_{Y|X^n=x^n} \stackrel{\left(k \epsilon, \frac{e^{k\epsilon} - 1}{e^{\epsilon} - 1} \delta \right)}{\approx} P_{Y|X^n=\tilde{x}^n}.
    \end{equation}
\end{property}

On the other hand, MI-DP does not have an analogous property.  Even if a mechanism satisfies $\epsilon$-MI-DP, there may not be a bound on $I(X_{\cal I};Y|X_{{\cal I}^c})$, where ${\cal I}$ represents a subset of $|{\cal I}|=k$ indices.  Consider the following example.

\begin{example}
\label{eg:group privacy}
    Consider a database with two entries, $X_1$ and $X_2$, which are real valued.  The randomized mechanism $P_{Y|X_1,X_2}$ produces an output $Y$ which can be a real number or one of two special values $e_1$ or $e_2$.  The behavior of the mechanism is best described in two cases:
    
    If $X_1 = X_2$:
    \begin{equation}
        Y = 
        \begin{cases}
        X_1, &\text{with probability } \epsilon, \\
        e_1, &\text{with probability } 1 - \epsilon.
        \end{cases}
    \end{equation}
    If $X_1 \neq X_2$:
    \begin{equation}
        Y = 
        \begin{cases}
        e_2, &\text{with probability } \epsilon, \\
        e_1, &\text{with probability } 1 - \epsilon.
        \end{cases}
    \end{equation}
    
    This mechanism satisfies $(\epsilon \ln 2)$-MI-DP.  Notice that for any value of $X_2=x_2$, we have a binary erasure channel from $X_1$ to $Y$, with binary input determined by whether $X_1 = x_2$ or not.  The symbol $e_1$ serves as the erasure.  The symbol $e_2$ represents the unerased indicator that $X_1 \neq x_2$.  This binary erasure channel with erasure probability $1-\epsilon$ has mutual information bounded above by $\epsilon \ln 2$ nats (the capacity of the erasure channel).
    
    On the other hand, the mutual information $I(X_1,X_2;Y)$ is unbounded if there are no constraints on ${\cal X}_1$ and ${\cal X}_2$.  Indeed, if we let $X_1$ be a continuous random variable, and we set $X_2 = X_1$, then
    \begin{equation}
        I(X_1,X_2;Y) = \infty.
    \end{equation}
    
    More generally, if the domains ${\cal X}_1$ and ${\cal X}_2$ are equal, then the capacity of the erasure channel gives the achievable mutual information (where $e_2$ represents an additional input symbol selected by any choice of $X_1 \neq X_2$):
    \begin{align}
        \max_{P_{X_1,X_2}} I(X_1,X_2;Y) &= \epsilon \log (|{\cal X}_1| + 1) \nonumber \\
        &= \epsilon \log (|{\cal Y}| - 1).
    \end{align}
\end{example}

In fact, Example~\ref{eg:group privacy} might be best interpreted as a fortunate advantage of MI-DP.  With any query mechanism, there is a trade-off between privacy and the informational utility to be gained from the output.  If we apply Property~\ref{property:group e} with $k=n$, the conclusion is that
\begin{equation}
    P_{Y|X^n=x^n} \stackrel{(n \epsilon, 0)}{\approx} P_{Y|X^n=\tilde{x}^n}
\end{equation}
for any two databases $x^n$ and $\tilde{x}^n$.  By revisiting the proof of Lemma~\ref{lem:e implies mi}, we obtain
\begin{equation}
    I(X^n;Y) \leq \min \left\{ n \epsilon, (n \epsilon)^2 \right\} \text{ nats}.
\end{equation}
One way to view this is as a crude bound on the utility of the query output.  The bound is detrimental if $n$ is not large.  On the other hand, MI-DP does not imply such a constraint.

If, however, we take into account a cardinality bound on the database entries or the query output, then there is indeed an upper bound on the information leaked from a group of database entries.  This is obtained by using Property~\ref{property:group d} in combination with Lemma~\ref{LEM:MI IMPLIES D}, followed by repeating the proof of Lemma~\ref{LEM:D IMPLIES MI} for a group rather than an individual entry.

\begin{cor}
    Suppose the randomized mechanism $P_{Y|X^n}$ satisfies $\epsilon$-MI-DP.  Then for any subset of indices ${\cal I}$, with $|{\cal I}| = k$, and with ${\cal I}^c = [n] \setminus {\cal I}$,
    \begin{equation}
        \sup_{P_{X^n}} I(X_{\cal I};Y|X_{{\cal I}^c}) \leq 2 h \left( k \sqrt{2 \epsilon} \right) + 2 k \sqrt{2 \epsilon} \log M,
    \end{equation}
    where $M = \min \left\{ |{\cal Y}|, \left( \max_i |{\cal X}_i| \right)^k + 1 \right\}$.
\end{cor}

\section{Variations of Diff. Privacy}
\label{sec:viewpoint}

Many variations of differential privacy have been proposed in the literature to provide different assurances.  Here we demonstrate how mutual-information differential-privacy can be adapted to correspond to these various definitions.

\subsection{Personalized Differential Privacy}
Personalized differential privacy \cite{jorgensen:personalizeddp} addresses the situation where participants of the database may have different concerns about the level of privacy.  This is handled by assigning a different $\epsilon_i$ for each database entry $X_i$.  That is, for any database instances $x^n$ and $\tilde{x}^n$ which differ in only the $i$th place,
\begin{equation}
    P_{Y|X^n=x^n} \stackrel{(\epsilon_i,0)}{\approx} P_{Y|X^n=\tilde{x}^n}.
\end{equation}

The modification to MI-DP would be to require that for each $i$,
\begin{equation}
   \sup_{P_{X^n}} I(X_i;Y|X^{-i}) \leq \epsilon_i \text{ nats}.
\end{equation}

\subsection{Free-Lunch Privacy}
Free-lunch privacy was both defined and refuted in \cite{kifer:nfl} as a stronger privacy definition which puts no restriction on which database instances must be indistinguishable.  A mechanism $P_{Y|X^n}$ is $\epsilon$-free-lunch private if every pair of database instances $x^n$ and $\tilde{x}^n$ satisfies
\begin{equation}
    P_{Y|X^n=x^n} \stackrel{(\epsilon,0)}{\approx} P_{Y|X^n=\tilde{x}^n}.
\end{equation}

The MI-DP equivalent of this would be
\begin{equation}
\label{eq:flp}
   \sup_{P_{X^n}} I(X^n;Y) \leq \epsilon \text{ nats}.
\end{equation}
We can easily see the strength of this definition by applying the chain rule of mutual information (Property~\ref{property:mi chain rule}) to \eqref{eq:flp}.  The result is that for any pair of disjoint index sets ${\cal I}$ and ${\cal J}$,
\begin{equation}
\label{eq:flp conditional}
    \sup_{P_{X^n}} I(X_{\cal I};Y|X_{\cal J}) \leq \epsilon \text{ nats}.
\end{equation}
On the other hand, \eqref{eq:flp} and \eqref{eq:flp conditional} illustrate the poor utility provided by the $\epsilon$-free-lunch privacy mechanism, as the information contained in the output is always upper bounded by $\epsilon$ regardless of distribution and prior knowledge.

\subsection{Bayesian Differential Privacy}
Bayesian differential privacy \cite{yang:bayesian} deals with the possible privacy degradation of differential privacy if the entries in the database are correlated.  As was discussed in Section~\ref{section:strong adversary}, a weak adversary who has less background knowledge of the database may stand to gain much more information than the adversary who knows all but one entry.  Bayesian differential privacy is meant to protect simultaneously against all adversaries, but in order to do so it assumes a prior distribution on the database.

Given a prior distribution $P_{X^n}$, a mechanism $P_{Y|X^n}$ is $\epsilon$-Bayesian differentially private if, for any index $i$ and subset of indices ${\cal I}$,
\begin{equation}
\label{eq:bayesian privacy}
    P_{Y|X_i=x_i,X_{\cal I} = x_{\cal I}} \stackrel{(\epsilon,0)}{\approx} P_{Y|X_i=\tilde{x}_i,X_{\cal I} = x_{\cal I}}.
\end{equation}
Notice that the conditional distributions are not necessarily conditioned on the entire database.

The MI-DP equivalent is, for any index $i$ and subset of indices ${\cal I}$,
\begin{equation}
    I(X_i;Y|X_{\cal I}) \leq \epsilon \text{ nats},
\end{equation}
which is in fact implied by \eqref{eq:bayesian privacy}.  Furthermore, this notion of privacy can be strengthened by maximizing over database distributions, making it a stronger notion of privacy than differential privacy.

In spite of the additional strength of this privacy metric (especially when removing the Bayesian prior assumption by maximizing over the database distribution), this is not nearly as pessimistic as free-lunch privacy.  As a comparison, the chain rule of mutual information (Property~\ref{property:mi chain rule}) in this case implies that for any two disjoint index sets ${\cal I}$ and ${\cal J}$,
\begin{equation}
    I(X_{\cal I};Y|X_{\cal J}) \leq |{\cal I}| \epsilon \text{ nats}.
\end{equation}
Consequently,
\begin{equation}
    I(X^n;Y) \leq n \epsilon \text{ nats}.
\end{equation}

\subsection{Adversarial Privacy}
Adversarial privacy \cite{rastogi:adversarial} does three things differently from differential privacy.  First, it assumes a prior distribution $P_{X^n}$ on the database (like Bayesian differential privacy).  Second, it does not restrict attention to neighboring database instances (like free-lunch privacy).  Third, it asymmetrically requires
\begin{equation}
\label{eq:adversarial privacy}
    \ln \frac{d P_{X^n|Y=y}}{d P_{X^n}} \leq \epsilon \quad \forall y \in {\cal Y}.
\end{equation}
The idea is that the adversary can not increase certainty about a particular database value by much, even while other database values may be eliminated.

Since mutual information is the expected value of the quantity on the left of \eqref{eq:adversarial privacy}, it is clear that adversarial privacy implies
\begin{equation}
    I(X^n;Y) \leq \epsilon \text{ nats}.
\end{equation}
Thus, adversarial privacy has similarity to free-lunch privacy, though in the Bayesian setting.  The subtleties of the asymmetric constraint are not captured in this MI-DP variant.

\section{R\'{e}nyi Entropy Generalization}
The notion of $\alpha$-mutual-information is the generalization of mutual information using R\'{e}nyi information measures. There are many proposed ways to accomplish such a generalization. Here we adopt Sibson's proposal (see \cite{verdu:alpha}):
\begin{equation}
    I_\alpha^\mathsf{s}(X;Y) = \min_{Q_Y} D_\alpha(P_{Y|X} \| Q_Y | P_X),
\end{equation}
where $D_\alpha$ is the conditional R\'{e}nyi divergence of order $\alpha$ and the minimization is over all distributions $Q_Y$ on ${\cal Y}$.  Shannon's mutual information corresponds to $\alpha = 1$.

A simple upper bound holds for $\alpha$-mutual information for all $\alpha \geq 0$, given in the following lemma.

\begin{lem}[$\alpha$-mutual-information upper bound]
\label{lem:alpha mi}
If a randomized mechanism $P_{Y|X^n}$ satisfies $\epsilon$-DP, then for all $\alpha \geq 0$, $i$, and instances of the remainder of the database $x^{-i}$,
\begin{equation}
    \sup_{P_{X_i}} I_\alpha^\mathsf{s}(X_i;Y|\xmi = x^{-i}) \leq \epsilon \text{ nats}.
\end{equation}
\end{lem}

\begin{proof}
Let $\alpha \geq 0$, $i$, and $P_{X^n}$ be arbitrary.  In order to abbreviate notation, denote the event $\{X^{-i} = x^{-i}\}$ as $U$. Pick an arbitrary $x_i \in {\cal X}_i$,
\begin{align}
    I_\alpha^\mathsf{s}(X_i;Y|U) &= \min_{Q_Y} D_\alpha(P_{Y|X_i, U} \| Q_Y | P_{X_i|U}) \nonumber \\
    &\leq   D_\alpha(P_{Y|X_i, U} \| P_{Y|X_i = x_i, U} | P_{X_i|U}) \nonumber \\
    &= D_\alpha(P_{Y|X_i, U} P_{X_i|U}\| P_{Y|X_i = x_i, U} P_{X_i|U}) \nonumber \\
    &= \frac{1}{\alpha - 1} \log \E\left[\frac{d P_{Y|X_i, U} P_{X_i|U}}{d P_{Y|X_i = x_i, U} P_{X_i|U}}(Y^*,X^*)\right]^{\alpha - 1} \nonumber \\
    &= \frac{1}{\alpha - 1} \log \E\left[\frac{d P_{Y|X_i, U}}{d P_{Y|X_i = x_i, U}}(Y^*,X^*)\right]^{\alpha - 1},
\end{align}
where $(Y^*, X^*) \sim P_{Y|X_i, U} P_{X_i|U}$.

For $\alpha \neq 1$, the $\epsilon$-DP constraint implies that $\frac{d P_{Y|X_i U}}{d P_{Y|X_i = x_i, U}} \leq e^{\epsilon}$ for all values of $X_i$, thus
\begin{align}
    I_\alpha^\mathsf{s}(X_i;Y|U) &\leq \frac{1}{\alpha - 1} \log \E[e^\epsilon]^{\alpha - 1} \nonumber \\
    &= \epsilon \text{ nats}.
\end{align}
For the case $\alpha = 1$, the $\alpha$-mutual-information reduces to Shannon's mutual information, and we have $I_\alpha^\mathsf{s}(X_i;Y|U) = I(X_i;Y|U) \leq \epsilon$ from Lemma~\ref{lem:e implies mi}.
\end{proof}

Combining Property~\ref{property:group e} with the proof of Lemma~\ref{lem:alpha mi} gives the following corollary:

\begin{cor}
If the mechanism $P_{Y|X^n}$ satisfies $\epsilon$-DP, then
\begin{equation}
    \sup_{P_{X^n}} I_\alpha^\mathsf{s}(X^n;Y) \leq n \epsilon \text{ nats}.
\end{equation}
\end{cor}

Furthermore, for $\alpha>0$, when maximizing over database distributions $P_{X^n}$, all three notions of $\alpha$-mutual-information discussed in \cite{verdu:alpha} are equivalent.  Thus,
\begin{equation}
\label{eq:alpha equivalent}
    \sup_{P_{X^n}} I_\alpha^\mathsf{s}(X^n;Y)
    =  \sup_{P_{X^n}} I_\alpha^\mathsf{a}(X^n;Y)
    =  \sup_{P_{X^n}} I_\alpha^\mathsf{c}(X^n;Y)
    \leq n \epsilon \text{ nats}.
\end{equation}
In \cite{alvim:min} and \cite{barthe:min}, the information leakage is defined as
\begin{equation}
    I_{\infty}(X^n;Y) = H_\infty(X^n) - H_\infty(X^n|Y)
\end{equation}
where
\begin{equation}
H_\infty(X^n|Y) = -\log \mathbb{E} \left[ \max_{x^n} P_{X^n|Y}(x^n|Y) \right].
\end{equation}
This definition matches Arimoto's proposal $I_\infty^\mathsf{a}(X^n;Y)$, so it is a special case of \eqref{eq:alpha equivalent}.

\section{Acknowledgements}
This work was supported by the Air Force Office of Scientific Research (grant FA9550-15-1-0180) and the National Science Foundation (grant CCF-1350595).

\bibliographystyle{abbrv}
\bibliography{ref}

\appendix

\section{Proof of Property~\ref{PROPERTY:E TO KL}}
\label{PROOF:E TO KL}

Assume that
\begin{equation}
    P \stackrel{(\epsilon, 0)}{\approx} Q.
\end{equation}
As stated in \eqref{eq:e0 restatement}, this gives
\begin{equation}
    \left| \ln \frac{dP}{dQ}(a) \right| \leq \epsilon \quad \forall a \in \Omega,
\end{equation}
which is equivalent to
\begin{equation}
    \frac{dP}{dQ}(a) \in \left[ e^{-\epsilon}, e^{\epsilon} \right] \quad \forall a \in \Omega.
\end{equation}

Consider that
\begin{align}
    D(P\|Q) &= \int dP(a) \ln \frac{dP}{dQ}(a) \nonumber \\
    &= \int dQ(a) \frac{dP}{dQ}(a) \ln \frac{dP}{dQ}(a) \nonumber \\
    &= \mathbb{E} \left[ \frac{dP}{dQ}(X) \ln \frac{dP}{dQ}(X) \right],
\end{align}
where $X \sim Q$.

Let us define the random variable $Z = \frac{dP}{dQ}(X)$.  We know the following facts:
\begin{align}
    Z &\in \left[ e^{-\epsilon}, e^{\epsilon} \right] \quad \text{w.p. } 1, \\
    \mathbb{E} [Z] &= 1, \\
    D(P\|Q) &= \mathbb{E} [Z \ln Z].
\end{align}

Since the function $f(x) = x \ln x$ for $x > 0$ is convex, we know that a distribution of $Z$ that maximizes $D(P\|Q)$ under these constraints places all mass at the endpoints of the allowed support interval.  Therefore, maximum $D(P\|Q)$ occurs with the following choice of distribution for $Z$:
\begin{equation}
    Z = 
    \begin{cases}
    e^{\epsilon}, & \text{w.p. } \frac{1 - e^{-\epsilon}}{e^{\epsilon} - e^{-\epsilon}}, \\
    e^{-\epsilon}, & \text{w.p. } \frac{e^{\epsilon} - 1}{e^{\epsilon} - e^{-\epsilon}}.
    \end{cases}
\end{equation}
A computation of $\mathbb{E} [Z \ln Z]$ gives the desired result.

This extreme is achieved by a symmetric pair of binary distributions, consistent with the distribution of $Z$ derived above.  Thus, coincidentally, for this choice of extreme distributions that maximize $D(P\|Q)$, it turns out that $D(P\|Q) = D(Q\|P)$.

The relaxation in Property~\ref{PROPERTY:E TO KL} can be arrived at by making the following observation:
\begin{align}
    \epsilon \frac{\left( e^{\epsilon} - 1 \right) \left( 1 - e^{-\epsilon} \right)}{\left( e^{\epsilon} - 1 \right) + \left( 1 - e^{-\epsilon} \right)} &\leq \epsilon \left( 1 - e^{-\epsilon} \right) \nonumber \\
    &\leq \min \left\{ \epsilon, \epsilon^2 \right\}. \label{eq:kl improve duchi}
\end{align}
Other bounds in the literature (Lemma~III.2 of \cite{dwork-rothblum-vadhan10} and Theorem~1 of \cite{duchi:lp}), while slightly loose, establish that $(\epsilon, 0)$-closeness implies an upper bound of roughly $\epsilon^2$ nats of Kullback-Leibler divergence for small $\epsilon$, which is only off by a factor of two.
\qed

\section{Proof of Property~\ref{PROPERTY:E TO D}}
\label{PROOF:E TO D}

Assume the $P$ and $Q$ are $(\epsilon, \delta)$-close.  To show that they are $(\epsilon', \delta')$-close, we must show that for any $A \in {\cal F}$
\begin{align}
    P(A) &\leq \delta' + e^{\epsilon'} Q(A), \label{proof p less q} \\
    Q(A) &\leq \delta' + e^{\epsilon'} P(A). \label{proof q less p}
\end{align}
By symmetry, we need only argue \eqref{proof p less q}.
    
We will build the proof from two inequalities.  The first is a direction application of \eqref{eq:ed p less q}:
\begin{equation}
    P(A) \leq \delta + e^{\epsilon} Q(A). \label{eq:repeat of p less q}
\end{equation}
The second is an application of \eqref{eq:ed q less p} to the complement of $A$, denoted as $A^c$:
\begin{equation}
    Q(A^c) \leq \delta + e^{\epsilon} P(A^c).
\end{equation}
By substituting $P(A^c) = 1 - P(A)$ and $Q(A^c) = 1 - Q(A)$ and rearranging, this implies
\begin{equation}
    P(A) \leq 1 - e^{-\epsilon}(1 - \delta) + e^{-\epsilon} Q(A). \label{eq:second bound p less q}
\end{equation}
    
Now we complete the proof with some simple manipulations and by substituting the value $\delta' = 1 - \frac{\left( e^{\epsilon'} + 1 \right)(1-\delta)}{e^{\epsilon}+1}$ stated in Property~\ref{PROPERTY:E TO D}.  From \eqref{eq:repeat of p less q} we can conclude
\begin{align}
    P(A) &\leq \delta + e^{\epsilon} Q(A) \nonumber \\
    &= \delta' + e^{\epsilon'} Q(A) + (\delta - \delta') + \left( e^{\epsilon} - e^{\epsilon'} \right) Q(A) \nonumber \\
    &= \delta' + e^{\epsilon'} Q(A) + \left( e^{\epsilon} - e^{\epsilon'} \right) \left( Q(A) - \frac{1 - \delta}{e^{\epsilon} + 1} \right). \label{eq:p bound small q}
\end{align}
From \eqref{eq:second bound p less q} we have
\begin{align}
    P(A) &\leq 1 - e^{-\epsilon}(1 - \delta) + e^{-\epsilon} Q(A) \nonumber \\
    &= \delta' + e^{\epsilon'} Q(A) + \left( 1 - e^{-\epsilon} (1 - \delta) - \delta' \right) + \left( e^{-\epsilon} - e^{\epsilon'} \right) Q(A) \nonumber \\
    &= \delta' + e^{\epsilon'} Q(A) + \left( e^{\epsilon'} - e^{-\epsilon} \right) \left( \frac{1 - \delta}{e^{\epsilon} + 1} - Q(A) \right). \label{eq:p bound large q}
\end{align}
If $Q(A) \leq \frac{1 - \delta}{e^{\epsilon} + 1}$ then \eqref{eq:p bound small q} establishes \eqref{eq:repeat of p less q}, since $\epsilon \geq \epsilon' \geq 0$.  Otherwise, \eqref{eq:p bound large q} establishes \eqref{eq:repeat of p less q}.
\qed

\section{Proof of Lemma~\ref{LEM:MI IMPLIES D}}
\label{PROOF:MI IMPLIES D}

According to the arguments immediately following Lemma~\ref{LEM:MI IMPLIES D}, we only need to show that the claim holds for randomized mechanisms $P_{Y|X}$ that have binary input and binary output.  That is, $|{\cal X}| = |{\cal Y}| = 2$.
    
Start by assuming that the randomized mechanism $P_{Y|X}$ satisfies $\epsilon$-MI-DP.  Since $X$ is a database with only one entry, $\epsilon$-MI-DP simply means
\begin{equation}
    \max_{P_X} I(X;Y) \leq \epsilon \text{ nats}.
\end{equation}
Notice that the left side is the expression for channel capacity from information theory, where $P_{Y|X}$ would be interpreted as a communication channel.  With this interpretation, what we are trying to show is that a bound on the channel capacity for binary channels implies a total variation bound between the conditional output distributions.  It has already been argued that binary channels contain the extreme cases, since total variation can be expressed as an inequality relating probabilities of a single arbitrary set (in general, the form of \eqref{eq:ed p less q} and \eqref{eq:ed q less p} gives this conclusion).  The next step is to show, specifically for total variation, that binary {\em symmetric} channels are the extreme cases.
    
Because $X$ and $Y$ are binary, the channel $P_{Y|X}$ can be parametrized with two parameters:
\begin{align}
    a &\triangleq \mathbb{P} [Y=1 | X=0], \\
    b &\triangleq \mathbb{P} [Y=1 | X=1].
\end{align}
Now consider the complementary channel $P_{\tilde{Y}|\tilde{X}}$, where $\tilde{X} = X \oplus 1$ and $\tilde{Y} = Y \oplus 1$, with $\oplus$ representing addition modulo 2.  This gives
\begin{align}
    \mathbb{P} \left[ \tilde{Y}=1 \middle| \tilde{X}=0 \right] &= 1 - b, \\
    \mathbb{P} \left[ \tilde{Y}=1 \middle| \tilde{X}=1 \right] &= 1 - a.
\end{align}
Finally, define a new binary channel, denoted as $P_{\hat{Y}|\hat{X}}$, which is a convex combination of the original channel and the complementary channel.  Then,
\begin{align}
    \mathbb{P} \left[ \hat{Y}=1 \middle| \hat{X}=0 \right] &= \frac{1}{2} + \frac{a - b}{2}, \\
    \mathbb{P} \left[ \hat{Y}=1 \middle| \hat{X}=1 \right] &= \frac{1}{2} - \frac{a - b}{2}.
\end{align}
    
Notice that for all three channels, $P_{Y|X}$, $P_{\tilde{Y}|\tilde{X}}$, and $P_{\hat{Y}|\hat{X}}$, the total variation between the two conditional output distributions is the same:
\begin{align}
    \left\| P_{Y|X=0} - P_{Y|X=1} \right\|_{TV} &= \left\| P_{\tilde{Y}|\tilde{X}=0} - P_{\tilde{Y}|\tilde{X}=1} \right\|_{TV} \nonumber \\
    &= \left\| P_{\hat{Y}|\hat{X}=0} - P_{\hat{Y}|\hat{X}=1} \right\|_{TV} \nonumber \\
    &= |a - b|.
\end{align}
On the other hand, channel capacity is a convex function of the channel parameters.  By symmetry, $P_{Y|X}$ and $P_{\tilde{Y}|\tilde{X}}$ have the same capacity.  Therefore, the convex combination $P_{\hat{Y}|\hat{X}}$, which is a binary {\em symmetric} channel, has a lower capacity.  Thus, binary symmetric channels are the extreme points in the trade-off between capacity and total variation.  For every binary channel, there is a binary symmetric channel with the same capacity but with greater or equal total variation distance between the conditional output distributions.
    
Finally, we arrive at Lemma~\ref{LEM:MI IMPLIES D} by applying the formula for channel capacity of a binary symmetric channel.  If we denote by $\delta$ the total variation distance between the conditional output distributions, then the cross-over probability is $\frac{1}{2} - \frac{\delta}{2}$.  The channel capacity is then
\begin{equation}
    C = \ln 2 - h \left( \frac{1}{2} - \frac{\delta}{2} \right) \text{ nats},
\end{equation}
where $h(\cdot)$ is the binary entropy function in nats.  Inverting this equation gives \eqref{eq:mi to d tight}.
    
The relaxed bound in \eqref{eq:mi to d simple} is established by the fact that the second order Tailor expansion of $h(x)$ about $x=\frac{1}{2}$ is in fact an upper bound:
\begin{equation}
    h(x) \leq \ln 2 - 2 \left( x - \frac{1}{2} \right)^2.
\end{equation}

An alternative simple argument directly arrives at the looser bound in \eqref{eq:mi to d simple}, without even reducing to the binary case.  We again refer to the geometric interpretation of capacity as the radius of the information ball \cite[Theorem~13.1.1]{cover:it}.  By Pinsker's inequality (Property~\ref{property:kl to d}), each conditional output distribution is within total variation distance $\sqrt{\frac{\epsilon}{2}}$ of the center of the information ball.  The triangle inequality gives \eqref{eq:mi to d simple}.
\qed

\section{Proof of Lemma~\ref{LEM:D IMPLIES MI}}
\label{PROOF:D TO MI}

Assume that the randomized mechanism $P_{Y|X^n}$ is $(\delta)$-DP, and let $i \in \{1,\ldots,n\}$ and $P_{X^n}$ be arbitrary.

Two proof arguments are needed, one based on the database entries $\{X_i\}$ having a finite set of possible values, and the other based on the same for the query response $Y$.  In both cases, however, we first note that the conditional mutual information $I(X_i;Y|X^{-i})$ is an expected value over instances of $X^{-i}$.  We provide bounds that uniformly hold for each instance of $x^{-i}$.  To that end, fix $x^{-i}$ arbitrarily, and let $(\tilde{X}, \tilde{Y}) \sim P_{X_i,Y|X^{-i}=x^{-i}}$.  This gives,
\begin{equation}
    I(X_i;Y|X^{-i}=x^{-i}) = I(\tilde{X};\tilde{Y}).
\end{equation}
Notice further that any two databases in the set $\{ \tilde{x}^n$ \; : \; $\tilde{x}^{-i} = x^{-i} \}$ are neighbors according to Definition~\ref{def:neighbor}.  Therefore, by assumption, \begin{equation}
\label{eq:tilde channel close other}
    P_{\tilde{Y}|\tilde{X}=\tilde{x}_1} \stackrel{(0,\delta)}{\approx} P_{\tilde{Y}|\tilde{X}=\tilde{x}_2}
\end{equation}
for any two values $\tilde{x}_1$ and $\tilde{x}_2$.

We now aim to bound $I(\tilde{X};\tilde{Y})$.  Consider first the case where $|{\cal Y}| < \infty$.  By construction, $|\tilde{\cal Y}| = |{\cal Y}|$.

From \eqref{eq:tilde channel close other} we can claim that for any value $\tilde{x}$
\begin{equation}
\label{eq:tilde channel close output}
    P_{\tilde{Y}|\tilde{X}=\tilde{x}} \stackrel{(0,\delta)}{\approx} P_{\tilde{Y}}.
\end{equation}
This is justified by letting $X' \sim P_{\tilde{X}}$ and noting
\begin{align}
    \left\| P_{\tilde{Y}|\tilde{X}=\tilde{x}} - P_{\tilde{Y}} \right\|_{TV} &= \left\| P_{\tilde{Y}|\tilde{X}=\tilde{x}} - \mathbb{E} \left[ P_{\tilde{Y}|\tilde{X}=X'} \right] \right\|_{TV} \nonumber \\
    &\leq \mathbb{E} \left[ \left\| P_{\tilde{Y}|\tilde{X}=\tilde{x}} - P_{\tilde{Y}|\tilde{X}=X'} \right\|_{TV} \right] \nonumber \\
    &\leq \delta, \label{eq:tilde channel close output justification}
\end{align}
where the first inequality is due to Jensen's inequality and the convexity of the total variation distance.

Next we decompose mutual information into entropy terms:
\begin{equation}
\label{eq:mi entropy decomposition Y}
    I(\tilde{X};\tilde{Y}) = H(\tilde{Y}) - H(\tilde{Y}|\tilde{X}).
\end{equation}

Finally, a continuity property of entropy found in \cite{zhang:entropy} (see (4) within), derived from optimal coupling and Fano's inequality, bounds the difference in entropy as a function of total variation distance and $|\tilde{\cal Y}|$.  Combining this with \eqref{eq:tilde channel close output} and \eqref{eq:mi entropy decomposition Y} gives
\begin{align}
    I(\tilde{X};\tilde{Y}) &\leq
    \begin{cases}
    h(\delta) + \delta \ln \left( |\tilde{\cal Y}| - 1 \right), & \delta \leq \frac{|\tilde{\cal Y}| - 1}{|\tilde{\cal Y}|}, \\
    \ln |\tilde{\cal Y}|, & \delta > \frac{|\tilde{\cal Y}| - 1}{|\tilde{\cal Y}|}
    \end{cases} \label{eq:d to mi tighter y} \\
    &\leq h(\delta) + \delta \ln |\tilde{\cal Y}| \text{ nats}.
\end{align}
    
Next we consider the case where $\max_i |{\cal X}_i| < \infty$.  By construction, $|\tilde{\cal X}| \leq \max_i |{\cal X}_i|$.

For this case, we take \eqref{eq:tilde channel close output} a bit further.  In fact,
\begin{equation}
\label{eq:tilde joint close independent}
    P_{\tilde{X},\tilde{Y}} \stackrel{(0,\delta)}{\approx} P_{\tilde{X}}P_{\tilde{Y}}.
\end{equation}
This is justified by letting $X' \sim P_{\tilde{X}}$ and noting
\begin{equation}
    \left\| P_{\tilde{X},\tilde{Y}} - P_{\tilde{X}} P_{\tilde{Y}} \right\|_{TV} = \mathbb{E} \left[ \left\| P_{\tilde{Y}|\tilde{X}=X'} - P_{\tilde{Y}} \right\|_{TV} \right].
\end{equation}

This time we decompose mutual information in the reverse direction:
\begin{equation}
\label{eq:mi entropy decomposition X}
    I(\tilde{X};\tilde{Y}) = H(\tilde{X}) - H(\tilde{X}|\tilde{Y}).
\end{equation}

To complete the proof we need a continuity argument for condition entropy.  The following lemma is inspired by ideas from \cite{boche:entropy} which in turn come from \cite{alicki-fannes04}.

\begin{lem}[Continuity of conditional entropy]
\label{lem:cont conditional entropy}
    If $P$ and $Q$ are two distributions on $\mathcal{U} \times \mathcal{V}$ with $|\mathcal{U}| < \infty$, then
    \begin{gather}
        \begin{split}
            P &\stackrel{(0, \delta)}{\approx} Q \\
            &\Downarrow \\
            \left| H_P(U|V) - H_Q(U|V) \right| &\leq 2 h \left( \frac{\delta}{\delta + 1} \right) + 2 \frac{\delta}{\delta + 1} \log |{\cal U}|.
        \end{split}
    \end{gather}
\end{lem}

\begin{proof}
    Since the bound in the lemma is monotonic in $\delta$, we assume without loss of generality that
    \begin{equation}
        \| P - Q \|_{TV} = \delta.
    \end{equation}
    
    We first translate closeness in total variation distance to the existence of a common distribution that is close to both relative to the boundaries of the set of probability distributions.  To be more precise, there exists a probability distribution $p^*$ which is a convex combination of $P$ and another probability distribution, with most of the convex weight on $P$, and the same relationship holds between $p^*$ and $Q$.  That is:
    \begin{align}
        p^* &= \frac{1}{1+\delta} P + \frac{\delta}{1+\delta} \hat{p} \label{eq:convex p} \\
        &= \frac{1}{1+\delta} Q + \frac{\delta}{1+\delta} \hat{q}. \label{eq:convex q}
    \end{align}
    Once we have established this existence, the exact construction of $p^*$, $\hat{p}$, and $\hat{q}$ will have no consequence on the conclusion.
    
    Consider the Hahn decomposition of the signed measure $P-Q$ into positive and negative parts that are mutually singular, represented by the non-negative measures $\mu^+$ and $\mu^-$, as follows:
    \begin{align}
        P - Q &= \mu^+ - \mu^-, \\
        \mu^+ &\geq 0, \\
        \mu^- &\geq 0, \\
        \mu^+ &\perp \mu^-.
    \end{align}
    The total measure of each part, $\mu^+$ and $\mu^-$, is the total variation between $P$ and $Q$, which is $\delta$.  Thus, to normalize $\mu^+$ and $\mu^-$ to become probability measures, we must divide by $\delta$.
    
    Let
    \begin{align}
        \hat{p} &= \frac{1}{\delta} \mu^-, \\
        \hat{q} &= \frac{1}{\delta} \mu^+.
    \end{align}
    Then we have that $p^*$ is the greater part of $P$ and $Q$, normalized as
    \begin{align}
        p^* &= \frac{1}{1 + \delta} \left( P + \mu^- \right) \\
        &= \frac{1}{1 + \delta} \left( Q + \mu^+ \right),
    \end{align}
    which satisfies both \eqref{eq:convex p} and \eqref{eq:convex q}.
    
    Next, to complete the proof, we show that
    \begin{align}
        \left| H_P(U|V) - H_{p^*}(U|V) \right| &\leq h \left( \frac{\delta}{\delta + 1} \right) + \frac{\delta}{\delta + 1} \log |\cal{U}|, \\
        \left| H_Q(U|V) - H_{p^*}(U|V) \right| &\leq h \left( \frac{\delta}{\delta + 1} \right) + \frac{\delta}{\delta + 1} \log |\cal{U}|.
    \end{align}
    By symmetry, an argument for only one of the inequalities is needed.
    
    The following bound is aided by defining a binary random variable $B \sim \rm{Bern}\left( \frac{\delta}{1 + \delta} \right)$ from which we construct $p^*_{U,V|B=0} = P$ and $p^*_{U,V|B=1} = \hat{p}$.  This has no effect on the marginal distribution of $p^*_{U,V}$.  We have
    \begin{align}
        H_{p^*}(U|V) &= H_{p^*}(U|V,B) + I_{p^*}(U;B|V) \nonumber \\
        &= \frac{1}{1 + \delta} H_P(U|V) + \frac{\delta}{1 + \delta} H_{\hat{p}}(U|V) + I_{p^*}(U;B|V).
    \end{align}
    Subtracting $H_P(U|V)$ from both sides gives
    \begin{align}
        &H_{p^*}(U|V) - H_P(U|V) \nonumber \\
        &= \frac{\delta}{1+\delta} \left( H_{\hat{p}}(U|V) - H_{p^*}(U|V) \right) + I_{p^*}(U;B|V).
    \end{align}
    Finally, the argument is completed by bounding the three non-negative terms.  The entropy terms are bounded by $\log |{\cal U}|$.  For the conditional mutual information, $I_{p^*}(U;B|V) \leq H_{p^*}(B) = h \left( \frac{\delta}{\delta + 1} \right)$.
\end{proof}

The proof of Lemma~\ref{LEM:D IMPLIES MI} is completed by applying Lemma~\ref{lem:cont conditional entropy} with $P_{\tilde{X},\tilde{Y}}$ as $P$ and $P_{\tilde{X}} P_{\tilde{Y}}$ as $Q$, due to \eqref{eq:tilde joint close independent}.  Combined with \eqref{eq:mi entropy decomposition X} this gives
\begin{align}
    I(\tilde{X};\tilde{Y}) &\leq 2 h \left( \frac{\delta}{\delta + 1} \right) + 2 \frac{\delta}{\delta + 1} \ln |\tilde{\cal X}| \label{eq:d to mi tighter x} \\
    &\leq 2 h(\delta) + 2 \delta \ln \left( |\tilde{\cal X}| + 1 \right) \text{ nats}. \hfill \qed
\end{align}

\balancecolumns
\end{document}